\documentclass[american,aps,pra,reprint,floatfix,superscriptaddress,longbibliography]{revtex4-1}
\usepackage[unicode=true,pdfusetitle, bookmarks=true,bookmarksnumbered=false,bookmarksopen=false, breaklinks=false,pdfborder={0 0 0},backref=false,colorlinks=false]{hyperref}
\hypersetup{colorlinks,linkcolor=myurlcolor,citecolor=myurlcolor,urlcolor=myurlcolor}
\usepackage{graphics,epstopdf,graphicx,amsthm,amsmath,amssymb,mathptmx,braket,colortbl,color,bm,framed,mathrsfs}
\usepackage[T1]{fontenc}
\usepackage[up]{subfigure}
\usepackage{tikz}
\usepackage[shortlabels]{enumitem}
\usepackage[normalem]{ulem}
\definecolor{myurlcolor}{rgb}{0,0,0.9}

\newcommand{\proj}[1]{| #1\rangle\!\langle #1 |}

\DeclareMathOperator{\trace}{Tr}
\newcommand{\Ptr}[2]{\trace_{#1}\Pa{#2}}
\newcommand{\Tr}[1]{\Ptr{}{#1}}

\newcommand{\Pa}[1]{\left[#1\right]}

\newcommand{\norm}[1]{\left\lVert #1 \right\rVert}

\DeclareMathOperator{\Var}{Var}
\DeclareMathOperator{\Cost}{Cost}
\DeclareMathOperator{\ad}{ad}
\def\Cat{\mathrm{Cat}}
\def\AC{\mathrm{AC}}
\def\CNOT{\mathrm{CNOT}}

\theoremstyle{plain}
\newtheorem{thm}{Theorem}
\newtheorem{lem}[thm]{Lemma}
\newtheorem{prop}[thm]{Proposition}
\newtheorem{cor}[thm]{Corollary}
\newtheorem{Def}[thm]{Definition}

\newtheorem{claim}[thm]{Claim}
 \newtheorem{example}{Example}

\newcommand*{\myproofname}{Proof}

\def\ot{\otimes}
\def\complex{\mathbb{C}}

\DeclareMathAlphabet{\mathcal}{OMS}{cmsy}{m}{n}

\makeatother

\begin{document}

\author{Lu Li}
\email{lilu93@zju.edu.cn}
\affiliation{Department of Mathematics, Zhejiang Sci-Tech University, Hangzhou, Zhejiang 310018, China}

\author{Kaifeng Bu}
\email{kfbu@fas.harvard.edu}
\affiliation{Department of Physics, Harvard University, Cambridge, Massachusetts 02138, USA}

\author{Dax Enshan Koh}
\email{dax\textunderscore koh@ihpc.a-star.edu.sg}
\affiliation{Institute of High Performance Computing, Agency for Science, Technology and Research (A*STAR), 1 Fusionopolis Way, \#16-16 Connexis, Singapore 138632, Singapore}

\author{Arthur Jaffe}
\email{arthur_jaffe@harvard.edu}
\affiliation{Department of Physics, Harvard University, Cambridge, Massachusetts 02138, USA}

\author{Seth Lloyd}
\email{slloyd@mit.edu}
\affiliation{Department of Mechanical Engineering, Massachusetts Institute of Technology, Cambridge, Massachusetts 02139, USA
}

\title{ Wasserstein Complexity of Quantum Circuits}

\begin{abstract}

Given a unitary transformation, what is the size of the smallest quantum circuit that implements it? This quantity, known as the quantum circuit complexity, is a fundamental property of quantum evolutions that has widespread applications in many fields, including quantum computation, quantum field theory, and black hole physics. In this letter, we obtain a new lower bound for the quantum circuit complexity in terms of a novel complexity measure that we propose for quantum circuits, which we call the quantum Wasserstein complexity. Our proposed measure is based on the quantum Wasserstein distance of order one (also called the quantum earth mover's distance), a metric on the space of quantum states. We also prove several fundamental and important properties of our new complexity measure, which stand to be of independent interest. Finally, we show that our new measure also provides a lower bound for the experimental cost of implementing quantum circuits, which implies a quantum limit on converting quantum resources to computational resources. Our results provide novel applications of the quantum Wasserstein distance and pave the way for a deeper understanding of the resources needed to implement a quantum computation.

\end{abstract}

\maketitle

\section{Introduction}

How many elementary quantum gates does it take to synthesize a desired unitary transformation $U$? Computing this number---called the quantum circuit complexity of $U$---is one of the holy grails of quantum computation, especially noisy intermediate-scale quantum computation where quantum resources are scarce~\cite{preskill2018quantum,bharti2022noisy}. By using fewer elementary gates, a quantum algorithm stands to take less time to produce its output; this in turn increases its prospect for achieving the highly sought goal of quantum-computational advantage~\cite{harrow2017quantum,dalzell2020how}.

In the search for optimal quantum circuits that use as few elementary gates as possible, one has taken various approaches to quantum circuit synthesis, including: reinforcement learning~\cite{fosel2021quantum}, quantum Karnaugh maps~\cite{bae2020quantum}, and ZX-calculus~\cite{duncan2020graph}. While these methods usually succeed at finding smaller circuits, they seldom achieve optimal circuit solutions~\footnote{Alternatively, one could aim to minimize the number of a particular type of gate in the circuit. For this goal, algorithms that achieve optimality---called resource-optimal circuit synthesis or re-synthesis algorithms \cite{mosca2021polynomial}---include those that minimize the T-count \cite{gosset2014lagrangian,amy2014polynomial,amy2019tcount,gheorghiu2021t, mosca2021polynomial} or CNOT-count \cite{gheorghiu2020reducing} of Clifford+T circuits. 
}. Computing the quantum circuit complexity of unitaries is generally hard, and no efficient algorithms are known~\cite{chia2021quantum,botea2018complexity}.

Unexpected connections  have been found between the notion of quantum circuit complexity and high energy physics. A solution to the wormhole-growth paradox has been proposed, which asks: in the anti-de Sitter/conformal field theory (AdS/CFT) correspondence, which quantity in the CFT is dual to the wormhole volume~\cite{susskind2016computational}? Stanford and Susskind conjecture that the answer is the quantum circuit complexity of the boundary state~\cite{stanford2014complexity}. In follow-up work by Brown and Susskind~\cite{brown2018second}, they conjecture that the quantum circuit complexity grows linearly for an exponentially long time. This was subsequently formalized and proved by Haferkamp et al.~\cite{haferkamp2022linear}. (See also \cite{li2022short} for two short proofs of this result.)

Nielsen pioneered a geometric approach to find upper and lower bounds for the quantum circuit complexity~\cite{nielsen2006geometric}. In a seminal series of papers, Nielsen and collaborators developed various geometric notions related to quantum computation and proved that the quantum circuit complexity of a unitary, up to polynomial factors and technical caveats, is equal to the circuit cost, defined as the length of the shortest path between two points in some curved Riemannian geometry~\cite{nielsen2006geometric, nielsen2006quantum, nielsen2006optimal,dowling2008geometry}. 
%These results were subsequently generalized to arbitrary $d$-level systems \cite{li2013geometry, luo2014geometry}.\KB{I don't think we need consider the generalized case}

A recent fruitful approach to bound Nielsen's circuit cost is to relate it to the quantum resources in a quantum circuit that refer to ingredients in the circuit responsible for a quantum speedup. A famous example of a quantum resource is nonstabilizerness, also known as magic. The Gottesman-Knill theorem proves that stabilizer circuits are efficiently classically simulable~\cite{gottesman1998heisenberg}, but when supplemented with magic states, such circuits become hard to simulate~\cite{jozsa2014classical, koh2017further, bouland2018complexity, yoganathan2019quantum}. Hence, magic could be thought of as a resource for a quantum speedup~\cite{bravyi2016trading,bravyi2019simulation,howard2017application,seddon2021quantifying,seddon2019quantifying,wang2019quantifying,koh2017computing,bu2019efficient,bu2022classical}. In recent works, the circuit cost was shown to be bounded below by the circuit's entangling power~\cite{eisert2021entangling}, magic~\cite{Bucomplexity22}, and sensitivity~\cite{Bucomplexity22}.

In this letter, we study the circuit complexity of quantum circuits using the so-called \textit{Wasserstein distance}. Classically, the Wasserstein distance is a metric on the space of probability distributions on a metric space that traces its origins to the works of Kantorovich~\cite{kantorovich1960mathematical}, Vaserstein~\cite{vaserstein1969markov}, and others. %The classical Wasserstein distance is a distance function between two probability distributions on a given metric space. 
The classical Wasserstein distance has numerous applications, including optimal transport~\cite{kantorovich1960mathematical,villani2009optimal,peyre2019computational}, image retrieval in computer vision~\cite{rubner2000earth}, and Wasserstein generative adversarial networks (WGANs) in classical machine learning~\cite{arjovsky17a,GoodfellowNIPS2014,GulrajaniNIPS2017}.

What is an appropriate quantum generalization of the classical Wasserstein distance that is useful for quantum computation?
To answer this question, several quantum generalizations of the Wasserstein distance have been proposed. These include the quantum $W_2$ distances, pioneered by
Carlen, Maas, Datta, Rouz\'e, Junge, and others,  based on a Riemannian metric on the manifold of quantum states \cite{Carlen2014analog,carlen2017gradient,carlen2020non,rouze2019concentration,datta2020relating,Junge2020,VanPRL2021} that generalize the classical Wasserstein distance of order 2. These $W_2$ distances have been shown to be related to the entropy and Fisher information, two important concepts in quantum information theory~\cite{datta2020relating}.
In addition, several candidates have been proposed that generalize the Wasserstein distance of order 1 \cite{chen2017matricia,ryu2018vector,PalmaIEEE2021}. In this letter, we focus on the quantum Wasserstein distance of order 1, proposed by De Palma, Marvian, Trevisan, and Lloyd~\cite{PalmaIEEE2021}, which can be regarded as a quantum version of the Hamming distance on $n$-qudit systems. This new version of the quantum Wasserstein distance has numerous applications in quantum information and quantum computation, such as quantum machine learning (where it is commonly referred to as the \textit{quantum earth mover's distance}) \cite{Kiani21}, quantum concentration inequalities \cite{palma2022quantum}, quantum differential privacy~\cite{hirche2022quantum}, characterizing limitations of variational quantum algorithms \cite{de2022limitations}, etc. However, little is hitherto known about the connection between the quantum Wasserstein distance and the quantum circuit complexity.

In this work, we introduce a complexity measure for quantum circuits, called the quantum Wasserstein complexity, which we define as the 
maximal distance between the input and output states of the circuit, measured using the quantum Wasserstein distance of order 1.  We show several useful properties of this complexity 
measure for quantum channels, such as subadditivity under concatenation, superadditivity
under tensor product (or additivity for correlation enhanced complexity measures). Most importantly, we show a connection between the quantum Wasserstein complexity and the circuit complexity, where we show that the quantum Wasserstein complexity provides a lower bound for the circuit cost.
Finally, we show that the quantum Wasserstein complexity also provides a lower bound on the 
experimental cost to implement quantum circuits.

\section{Main results}
Given the $n$-qudit system $\mathcal{H}=(\complex^d)^{\ot n}$, let us define $O(\mathcal{H})$ to be the set of traceless, Hermitian operators on $\mathcal{H}$, i.e., 
$O(\mathcal{H})=\set{A \in \mathcal L(\mathcal H):\Tr{A}=0,\ A^\dag=A}$, where $\mathcal L(\mathcal H)$ denotes the set of linear operators on $\mathcal H$. Also, let $D(\mathcal{H})$ denote the 
set of density operators on $\mathcal{H}$. 
The \textit{quantum Wasserstein distance of order 1} \cite{PalmaIEEE2021} between two quantum states $\rho,\sigma\in D(\mathcal{H})$
is defined as:
\begin{eqnarray}
\nonumber\norm{\rho-\sigma}_{W_1}
:=\frac{1}{2} \min_{\set{X_i}}
\left\{\sum_i\norm{X_i}_1: \rho-\sigma=\sum_iX_i\;,\right.\\
\left. \phantom{\sum_i}  X_i\in O(\mathcal{H}),\ \Ptr{i}{X_i}=0 \ \forall i \right\},
\label{eq:def_W1}
\end{eqnarray}
where $\Ptr{i}{\cdot}$ denotes the partial trace over the $i$-th subsystem, and $\norm{X}_1=\Tr{\sqrt{X^\dag X}}$ denotes the trace norm of $X$. This distance is also called the \textit{quantum} $W_1$ \textit{distance} or the \textit{quantum earth mover's distance}.
%(The distance induced by the trace norm is called trace distance here). 
Using the quantum Wasserstein distance lets one distinguish between two quantum states that differ locally; this is different from their global distinguishability.  
One example to show this distinction comes from considering the $n$-qubit states $\ket{0}^{\ot n}$, $\ket{1}\ket{0}^{\ot n-1}$ and $\ket{1}^{\ot n}$. The trace distances \footnote{Here, the trace distance $d$ is defined as one-half the metric induced by the trace norm, i.e. $d:(x,y)\mapsto \frac{1}{2}\norm{x-y}_1$.} between $\ket{0}^{\ot n}$ and $\ket{1}\ket{0}^{\ot n-1}$ and between $\ket{0}^{\ot n}$ and $\ket{1}^{\ot n}$ are both equal to $1$; hence, using the trace distance, we cannot distinguish which of  $\ket{1}\ket{0}^{\ot n-1}$ and $\ket{1}^{\ot n}$ is further away from $\ket{0}^{\ot n}$.  However, the quantum $W_1$ distance  between $\ket{0}^{\ot n}$ and $\ket{1}\ket{0}^{\ot n-1}$ is equal to $1$, while the quantum $W_1$ 
distance between $\ket{0}^{\ot n}$ and $\ket{1}^{\ot n}$ is equal to $n$; hence in quantum $W_1$ distance, $\ket{1}^{\ot n}$  is further away from $\ket{0}^{\ot n}$ than $\ket{1}\ket{0}^{\ot n-1}$.

Another example comes from considering the trace distance between $\ket{0}^{\ot n}$ and the Cat state $\ket{\Cat_a}=a\ket{0}^{\ot n}+\sqrt{1-|a|^2}\ket{1}^{\ot n}$ with $0<|a|<1$. The trace distance is $\sqrt{1-|a|^2}$, which is very small when $a$ is close to $1$. However, the experimental resource, e.g., the number of gates, 
to transform $\ket{0}^{\ot n}$ to $\ket{\Cat_a}$ is proportional to the size of the system $n$, independent of how close  $a$ is to $1$. On the other hand, the quantum $W_1$ distance between $\ket{0}^{\ot n}$ and $\ket{\Cat_a}$ is $\Omega((1-|a|^2)n)$. 
This motivates us to consider the circuit complexity in terms of the quantum $W_1$ distance.

\begin{Def}
Given an $n$-qudit quantum channel $\Lambda:D(\mathcal{H})\to D(\mathcal{H})$, the \textit{quantum Wasserstein complexity} $C_{W_1}(\Lambda)$ is the maximal distance between the input state and output state in quantum $W_1$ distance, 
\begin{eqnarray}
C_{W_1}(\Lambda)
:=\max_{\rho\in D(\mathcal{H})}
\norm{\rho-\Lambda(\rho)}_{W_1}.
\end{eqnarray}
\end{Def}
By the convexity of the quantum $W_1$ distance, we need only to take the maximization over all pure states, 
\begin{eqnarray}
C_{W_1}(\Lambda)
=\max_{\ket{\psi}\in \mathcal{H}}
\norm{\proj{\psi}-\Lambda(\proj{\psi})}_{W_1}.
\end{eqnarray}
To demonstrate applications of the quantum Wasserstein complexity, let us first study its basic properties.

\begin{prop}\label{prop:bas_1}
The quantum Wasserstein complexity $C_{W_1}(\Lambda) $ satisfies the following properties:

\begin{enumerate}[(1)]
    \item{} \label{Faithfulness} 
Faithfulness: $C_{W_1}(\Lambda)=0$ if and only if $\Lambda$ is the identity map;

\item{} \label{Convexity}
Convexity: $C_{W_1}\left(\sum_ip_i\Lambda_i\right)\leq \sum_ip_iC_{W_1}(\Lambda_i)$, where 
$p_i\geq 0$ and $\sum_ip_i=1$;

\item{} \label{Subadditivity}
Subadditivity under concatenation:  $C_{W_1}(\Lambda_1\circ \Lambda_2)\leq C_{W_1}(\Lambda_1)+C_{W_1}(\Lambda_2)$;

\item{} \label{SubTensorization}
$C_{W_1}(\Lambda_1\ot \Lambda_2)\leq C_{W_1}(\Lambda_1\ot I)+C_{W_1}(\Lambda_2\ot I)$;

\item{} \label{SuperTensorization}
Superadditivity under tensorization: $C_{W_1}(\Lambda_1\ot \Lambda_2)\geq C_{W_1}(\Lambda_1)+C_{W_1}(\Lambda_2)$;

\item{} \label{Unitary}
For a unitary channel $U$, $C_{W_1}(U)=C_{W_1}(U^\dag)$;

\item{} \label{k-qubits}
If the quantum channel $\Lambda$ acts nontrivially on a $k$-qudit subsystem, then we have 
$C_{W_1}(\Lambda)\leq k$.
\end{enumerate}
\end{prop}
We prove Proposition \ref{prop:bas_1} in Appendix~\ref{Appen:basic}.
Note that there exists a quantum channel $\Lambda$, such that the inequality \ref{SuperTensorization} holds strictly. For example, for the single-qudit depolarizing channel $D_p(\cdot)=p(\cdot)+(1-p)\Tr{\cdot}I/d$, it holds that 
$C_{W_1}(D_p)=(1-p)(1-1/d)$ and
$C_{W_1}(D_p\ot I)=(1-p)(1-1/d^2)$. Hence, the inequality in \ref{SuperTensorization} holds strictly. See more examples in Table~\ref{tab:examples}.

%\begin{widetext}
\begin{table*}
	%\begin{ruledtabular}
		\begin{tabular}{|c|c|c|}\hline
			quantum channel $\Lambda$& $C_{W_1}(\Lambda)$ &$\AC_{W_1}(\Lambda)$\\ \hline
			Single-qudit depolarizing channel $D_p$&~~$ (1-p)(1-1/d)$~~&$(1-p)(1-1/d^2)$  \\ \hline
		       ~~ Tensor product of depolarizing channels~~&   ~~$\geq n(1-p)(1-1/d)$  ~~    & $n(1-p)(1-1/d^2)$\\ 
			        $D^{\ot n}_p$     & $\leq n(1-p)(1-1/d^2)$  &                                 \\ \hline
			Tensor product of Hadamard gates $H^{\ot n}$&$n$&$n$\\ \hline
			CNOT gates $\prod_i \CNOT_{i,i+1}$&$\Theta(n)$&$\Theta(n)$\\\hline
		\end{tabular}
			\caption{\label{tab:examples} We calculate (or estimate) the quantum Wasserstein complexity $C_{W_1}$ and $\AC_{W_1}$ for 
			several examples (See details in Appendix \ref{appen:exam}). Here, $f(n)=\Theta(n)$ means that there exist constants $c_1,c_2$ such that $c_1n\leq f(n)\leq c_2 n$ for any $n$.}
	%\end{ruledtabular}
\end{table*}
%\end{widetext}

The effects of the environment on quantum systems
are usually taken into consideration, which means the ancila qudits should be considered.

\begin{Def}
Given an $n$-qudit quantum channel $\Lambda:D(\mathcal{H})\to D(\mathcal{H})$, the 
\textit{correlation-assisted quantum Wasserstein complexity} $\AC_{W_1}$
for an $n$-qudit quantum channel $\Lambda$ is defined as follows
\begin{eqnarray}\label{def:AC}
\AC_{W_1}(\Lambda)
:=\sup_{m\in \mathbb{Z}_{\geq 0}}
C_{W_1}(\Lambda\ot I_m),
\end{eqnarray}
where $\mathbb{Z}_{\geq 0}$ denotes the set of nonnegative integers, and $I_m$ denotes the $m$-qubit identity operator.
\end{Def}
Note that 
$C_{W_1}(\Lambda\ot I_m)\geq C_{W_1}(\Lambda\ot I_{m-1}), \forall m \in \mathbb{Z}_{\geq 0}$ and
$
C_{W_1}(\Lambda)
\leq n
$.
Hence, $\AC_{W_1}$ is well-defined.

 \begin{prop}\label{prop:bas_2}
The correlation-assisted Wasserstein complexity $\AC_{W_1}(\Lambda)$ 
satisfies the following properties:

\begin{enumerate}[(1)]
    \item{} \label{CAFaithfulness}
Faithfulness: $\AC_{W_1}(\Lambda)=0$ if and only if $\Lambda$ is the identity map;

\item{} \label{CAConvexity}
Convexity: $\AC_{W_1}(\sum_ip_i\Lambda_i)\leq \sum_ip_i\AC_{W_1}(\Lambda_i)$, where 
$p_i\geq 0$ and $\sum_ip_i=1$;

\item{} \label{CASubadditivity}
Subadditivity under concatenation:  $\AC_{W_1}(\Lambda_1\circ \Lambda_2)\leq \AC_{W_1}(\Lambda_1)+\AC_{W_1}(\Lambda_2)$;

\item{} \label{CATensorization}
Additivity under tensorization: $\AC_{W_1}(\Lambda_1\ot \Lambda_2)=\AC_{W_1}(\Lambda_1)+\AC_{W_1}(\Lambda_2)$.
\end{enumerate}
\end{prop}
The proof of Proposition \ref{prop:bas_2} is provided in Appendix~\ref{Appen:basic}.
Note that, unlike $C_{W_1}$, $\AC_{W_1}$ is additive under tensorization.
Based on the definition, it is easy to see that $\AC_{W_1}(\Lambda)\geq C_{W_1}(\Lambda)$. Also, there exists some quantum channel $\Lambda$, e.g., the $1$-qudit depolarizing channel, such that $\AC_{W_1}(\Lambda)>C_{W_1}(\Lambda)$.
(See the examples in Table \ref{tab:examples}).

After studying the basic properties of $C_{W_1}$ and $\AC_{W_1}$,
let us consider an application of the quantum Wasserstein complexity to the study of the circuit complexity of quantum circuits. 
The circuit complexity of a unitary operator $U$ is defined as the minimum number of basic gates needed to  generate  $U$ \cite{nielsen2010quantum,kitaev2002classical,aaronson2016complexity} (See Fig.~\ref{fig1}). Here, we consider the circuit cost of quantum circuits, which 
was introduced by Nielson et al.~\cite{nielsen2006geometric,nielsen2006optimal,nielsen2006quantum} as a geodesic distance
from the identity operator to the target unitary with respect to a given metric,
and was shown to provide a useful lower bound for the quantum circuit complexity.

\begin{figure}[t]
  \center{\includegraphics[width=7cm]  {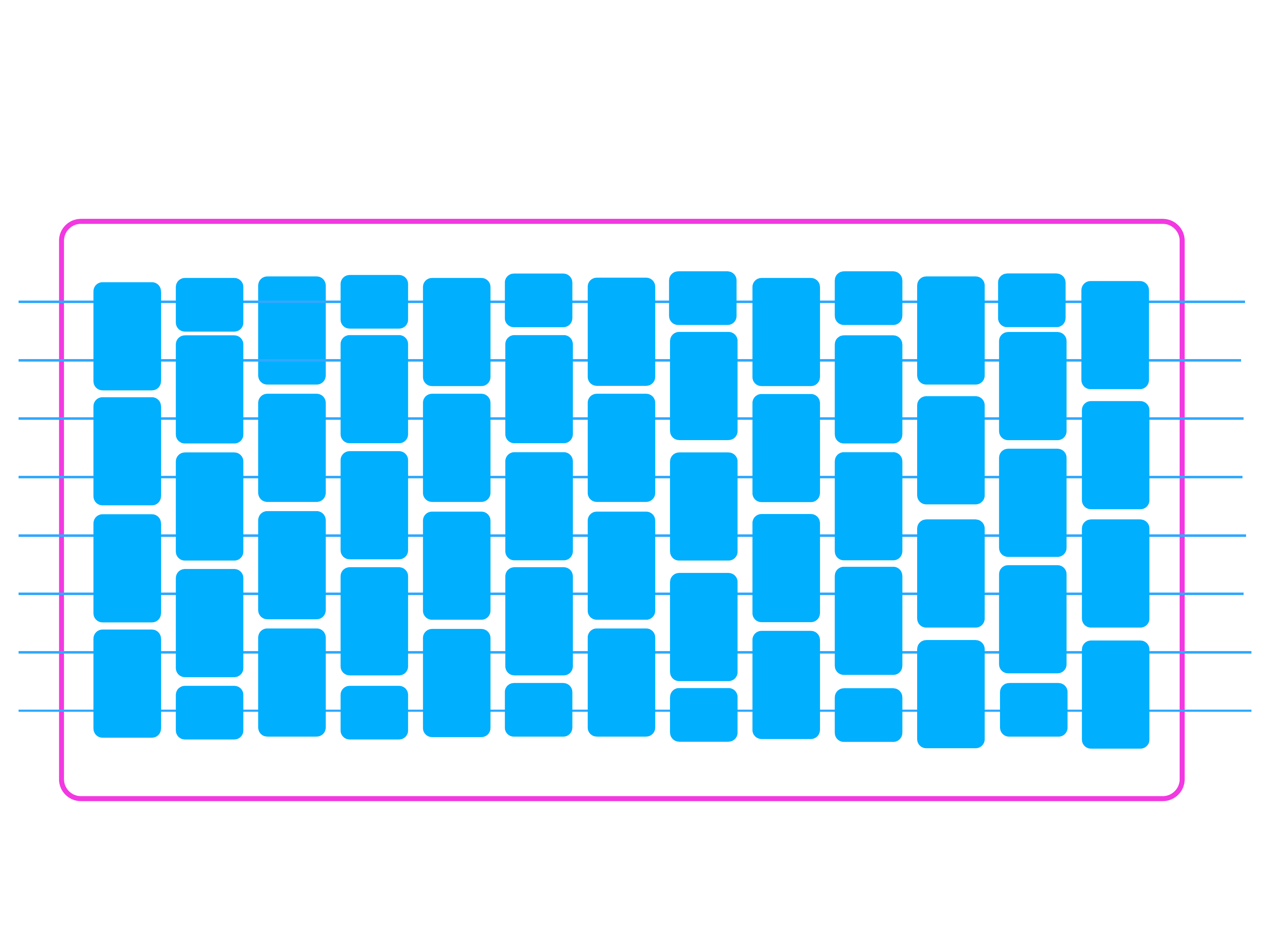}}     
  \caption{
  Diagram showing a quantum circuit composed of single-qubit and two-qubit basic gates. Quantum circuit complexity quantifies the minimum number of basic gates needed to generate a target unitary.}
  \label{fig1}
 \end{figure}

Here, the circuit cost of a unitary $U\in \mathrm{SU}(d^n)$ with respect to traceless Hermitian operators  $h_1, \ldots, h_m$, supported on 2 qudits and normalized as $\norm{h_i}_{\infty}=1$, is defined to be
\begin{eqnarray}
\label{eq:def_cost}
\Cost(U):=\inf\int^1_0
\sum^m_{j=1}|r_j(s)| d s,
\end{eqnarray}
where the infimum in \eqref{eq:def_cost} is taken over all continuous functions $r_j:[0,1]\to \mathbb R$ that satisfy
$
H(s)=
\sum^m_{j=1} r_j(s)h_j
$ and
$
U=\mathcal P\exp\left(-i\int^1_0H(s)ds\right),
\label{eq:path-orderedU}
$
where $\mathcal P$ denotes the path-ordering operator. (See Fig.~\ref{fig2})

\begin{figure}[t]
  \center{\includegraphics[width=7cm]  {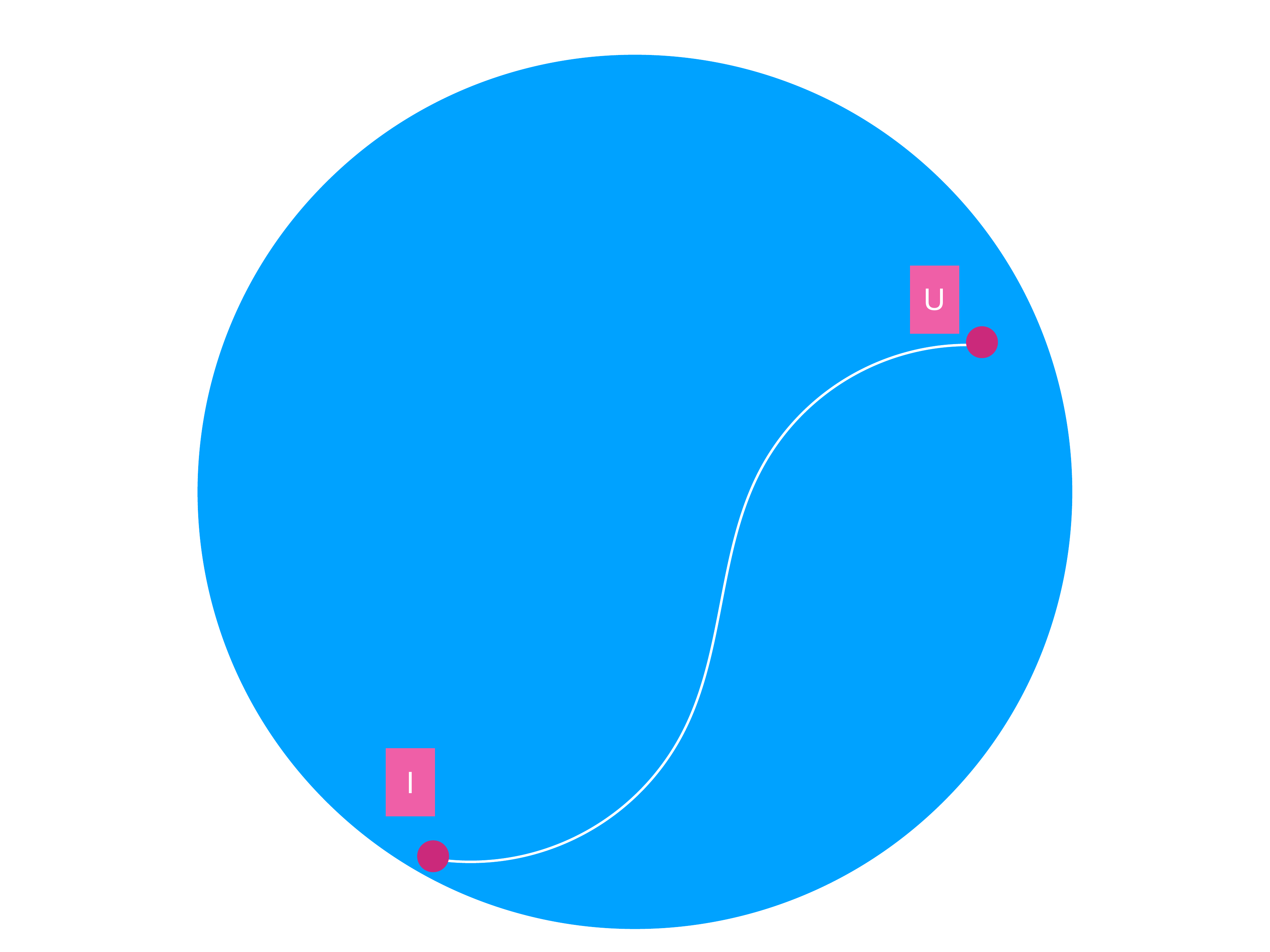}}     
  \caption{
  The circuit cost is the geodesic distance from the identity operator $I$ to the target unitary $U$ with respect to some metric, as illustrated by the length of the curve connecting $I$ and $U$ in the diagram.}
  \label{fig2}
 \end{figure}

For an $n$-qudit Hamiltonian acting nontrivially on a $k$-qudit subsystem, there is a simple upper bound on the total change of the quantum Wasserstein complexity through unitary evolution.
\begin{lem}
Given an $n$-qudit system with a $k$-qudit Hamiltonian $H$, 
the change of 
the quantum Wasserstein complexity for the unitary evolution $U_t=\exp(-i H t)$
with any time interval $\Delta t$ is bounded as follows
\begin{eqnarray}
C_{W_1}
(U_{t+\Delta t})-C_{W_1}(U_t)
\leq k.
\end{eqnarray}
\end{lem}
\begin{proof}
This comes from the fact that $C_{W_1}
(U_{t+\Delta t})=C_{W_1}
(U_{\Delta t}\circ U_{t})\leq C_{W_1}(U_t)+C_{W_1}(U_{\Delta t})$, and 
$C_{W_1}(U_{\Delta t})\leq k$.
\end{proof}
To quantify the change of the quantum Wasserstein complexity in an infinitesimally small time interval, we need
to introduce the quantum Wasserstein rate in the unitary dynamics 
generated by a Hamiltonian $H$, which will be useful for the connection
between the quantum Wasserstein complexity and the circuit complexity.
Given an $n$-qudit Hamiltonian $H$ and an
$n$-qudit pure state $\ket{\psi}$, the  quantum Wasserstein rate
of the unitary $U_t=\exp(-i H t)$ on the state $\ket{\psi}$ is defined as 
\begin{eqnarray}
R_{W_1}(\psi, H)
=\lim_{\Delta t\to 0}
 \norm{\frac{\proj{\psi}-U_{\Delta t}\proj{\psi} U^\dag_{\Delta t}}{\Delta t}}_{W_1}.
\end{eqnarray}

\begin{thm}[Small incremental quantum Wasserstein complexity]\label{thm:WCom_rate}
Given an $n$-qudit system with the
Hamiltonian $H$ acting on a $k$-qudit subsystem, and an $n$-qudit state 
$\ket{\psi}$, one has
\begin{eqnarray}
R_{W_1}(\psi, H)\leq 2\sqrt{2}k\norm{H}_{\infty},
\end{eqnarray}
where $\norm{\cdot}_{\infty}$ denotes the operator norm.
\end{thm}
We prove Theorem \ref{thm:WCom_rate} in  Appendix \ref{Appen:appl}.
The small incremental quantum Wasserstein complexity is a key point to show the connection between circuit cost and the quantum Wasserstein complexity, for which we have the following relationship.

\begin{thm}\label{thm:main1}
Given a unitary $U$, the circuit cost $\Cost(U)$ is lower bounded by the quantum Wasserstein complexity
$C_{W_1}(U)$ as
\begin{eqnarray}
\Cost(U)\geq 4\sqrt{2}C_{W_1}(U).
\end{eqnarray}
\end{thm}
We prove Theorem \ref{thm:main1} in Appendix \ref{Appen:appl}.
This relation can be generalized to the correlation-assisted quantum  Wasserstein complexity $\AC_{W_1}$ defined in \eqref{def:AC}.
\begin{cor}
Under the assumptions of  Theorem~\ref{thm:main1}, we have 
\begin{eqnarray}
\Cost(U)\geq 4\sqrt{2}\AC_{W_1}(U).
\end{eqnarray}
\end{cor}

Besides the circuit cost, another important complexity measure of 
quantum circuits is the experimental cost to implement the circuit. Let us consider an $n$-qubit quantum circuit with a
 sequence of quantum gates
$U=\prod_l U_l$, $U_l=\exp(-iH_lT_l)$, i.e., 
we run the gate $U_1$ first, then $U_2$, and so on. $T_l$ is the runtime of the $l$-th gate, and for each $l$, the time-independent Hamiltonian $H_l$ acting on $k_l$ qubits has the following spectral decomposition:
$
H_l=\sum^{2^{k_l}}_{i=1}
h_i\proj{h_i}, h_{i}\geq h_{i-1}, \forall i
$. 

Girolami and Anz\`a~\cite{GirolamiPRL2021}  proposed using the seminorm $E_l:=(h_{2^{k_l}}-h_1)/2$ defined by the Hamiltonian $H_l$ \cite{BoixoPRL2007} to define the experimental cost of implementing the gate $U_l$ and the circuit $U$ in terms of physical resources as 
\begin{eqnarray}
\mathcal{R}_{U_l}:=k_lE_lT_l,\quad
\mathcal{R}_U
:=\sum_l\mathcal{R}_{U_l}.
\end{eqnarray}
Previously, the experimental cost of the quantum circuits was investigated using a weighted version of the Bures distance \cite{GirolamiPRL2021}. Here
we find a new connection between the
experimental cost $\mathcal{R}_U$ and the quantum Wasserstein complexity  
$C_{W_1}(U)$ as follows:

\begin{thm}\label{thm:exp_cost}
The experimental cost $\mathcal{R}(U)$ of an $n$-qubit quantum circuit $U=\prod_l U_l$ is bounded from below by the 
quantum Wasserstein complexity $C_{W_1}(U)$ as
\begin{eqnarray}\label{eq:main_exp}
\mathcal{R}(U)
\geq \frac{1}{2}C_{W_1}(U).
\end{eqnarray}
\end{thm}

The proof of Theorem~\ref{thm:exp_cost} is presented in Appendix~\ref{Appen:appl}.
A similar argument also works for $\AC_{W_1}$ defined in \eqref{def:AC}. 
\begin{cor}
Under the same conditions as in Theorem~\ref{thm:exp_cost}, we have 
\begin{eqnarray}
\mathcal{R}(U)
\geq \frac{1}{2}\AC_{W_1}(U).
\end{eqnarray}
\end{cor}

\begin{example}\rm
Consider the CNOT gate $\exp(iH\pi/2)$, where 
$H=\frac{1}{2}(I-Z_1-X_2+Z_1 X_2)$. The experimental cost of CNOT is $\mathcal{R}(CNOT)=\pi$, and thus  the quantum circuit 
$U=\prod_iCNOT_{i,i+1}$, which can be used to generate the cat state $(\ket{0}^{\ot n}+\ket{1}^{\ot n})/\sqrt{2}$, has the 
experimental cost $\mathcal{R}(U)=\pi n$. The quantum Wasserstein complexity is $C_{W_1}(U)=\Theta(n)$. (See Table~\ref{tab:examples}.) Thus, the quantum circuit $U=\prod_iCNOT_{i,i+1}$ has an experimental cost which is equivalent to the quantum Wasserstein complexity, up to some constant factor. Also relation \eqref{eq:main_exp} holds.

\end{example}

\section{Conclusion}
We introduce a new measure of quantum complexity, which we call \textit{quantum Wasserstein complexity}. We provide a  connection between this complexity and the circuit cost. This provides a useful lower bound 
for the circuit complexity of the shallow (or constant-depth) quantum circuits, in terms of the quantum Wasserstein complexity.
We show that the quantum Wasserstein complexity provides a lower bound for the experimental 
cost of implementing  quantum circuits.

Our results provide an application and operational interpretation of the quantum Wasserstein 
distance for shallow quantum circuits. 
This raises the interesting  question: can one find a better bound for the circuit complexity of  deep quantum circuits based on the quantum Wasserstein distance? To compute the quantum Wasserstein complexity, one needs to maximize over all pure states, making it difficult to compute. So one naturally asks: can one discover an efficient way to approximate the quantum Wasserstein complexity? It is also interesting to study quantum Wasserstein distances of order other than 1, and how they relate to circuit complexity in quantum computation.

% Also, the hardness of finding the minimal size of quantum circuit of a given unitary is still unknown. Nevertheless, there is some evidence that this problem 
% is hard. If you can design an efficient (quantum) algorithm for UMCSP, then you can solve some frontier open problems in pseudorandomness and   cryptography, 
% e.g., the nonexistence of pseudorandom unitaries~\cite{chia2021quantum}. Hence, it is challenging

% {\color{blue} Seth's comment:
%  I wonder about the question in the end of the paper about how hard it is to 
% calculate this quantity.    Quantum Wasserstein-1 distance is itself not so easy to calculate: you can do it by
% linear programming, but because the dimension of the Hilbert space grows exponentially in the number of qudits, 
% this is still hard.
% In fact, I suspect that finding the size of the minimal circuit to produce a particular unitary is NP complete (or, rather, QMA complete). In particular, bounding this size from below should be hard: but then your elegant result implies that finding the Quantum Wasserstein 1 distance is also hard.

% }

% \KF{
% The hardness of finding the minimal size of quantum circuit for a given unitary (i.e., UMSCP) is still unknown. But there is some evidence that this problem 
% is hard. 
% If you  can design an efficient (quantum) algorithm for UMCSP, then you can solve some frontier open problems in Pseudorandomness and   cryptography, 
% e.g., the nonexistence of pseudorandom unitaries. 

\begin{acknowledgments}
We thank  Michael Freedman,  Chi-Ning Chou and Xun Gao for discussion of  quantum circuit complexity;
we thank Marius Junge, Weichen Gu and Hao Zhong for discussion of the quantum Wasserstein distance.
This work was supported in part by the ARO Grant W911NF-19-1-0302 and the ARO MURI Grant W911NF-20-1-0082.
\end{acknowledgments}

 \bibliography{cost-lit}
\clearpage
\newpage

\appendix
\widetext

\section{Basic properties of quantum Wasserstein norm of order 1}

We state three basic properties of the quantum Wasserstein norm of order 1~\cite{PalmaIEEE2021}, which is defined for all traceless, Hermitian operators $A$:
\begin{eqnarray}
\nonumber\norm{A}_{W_1}
:=\frac{1}{2} \min_{\set{X_i}}
\left\{\sum_i\norm{X_i}_1: A=\sum_i X_i,\ X_i\in O(\mathcal{H}),\ \Ptr{i}{X_i}=0 \ \forall i \right\}\;.
\label{eq:def_W1_norm}
\end{eqnarray}

\begin{lem}[De Palma et al.~\cite{PalmaIEEE2021}]  \label{lem:ten}(Tensorization)
Given an $n$-qudit traceless, Hermitian  operator $A$, and a subset $S\subset [n]$ with complement $S^c$, 
\begin{eqnarray}
\norm{A}_{W_1}
\geq 
\norm{\Ptr{S}{A}}_{W_1}
+\norm{\Ptr{S^c}{A}}_{W_1}\;.
\end{eqnarray}
\end{lem}

\begin{lem}[De Palma et al.~\cite{PalmaIEEE2021}] \label{lem:boun_1norm}
Given an $n$-qudit traceless, Hermitian operator $A$ and a subset $S\subset[n]$ with $\Ptr{S}{A}=0$,
\begin{eqnarray}
\norm{A}_{W_1}
\leq |S|
\frac{d^2-1}{d^2}
\norm{A}_1.
\end{eqnarray}
Hence, for any $n$-qudit quantum channel $\Lambda$ acting nontrivially on $k$-qudit subsystems, 
\begin{eqnarray}
\norm{\rho-\Lambda(\rho)}_{W_1}
\leq k.
\end{eqnarray}
\end{lem}

\begin{lem}[De Palma et al.~\cite{PalmaIEEE2021}]
Given an $n$-qudit traceless, Hermitian operator $A$, we have 
\begin{eqnarray}
\frac{1}{2}\norm{A}_1
\leq \norm{A}_{W_1}
\leq \frac{n}{2}\norm{A}_1.
\end{eqnarray}
Moreover, if there exists some $i$ such that $\Ptr{i}{A}=0$, then we have 
\begin{eqnarray}
\norm{A}_{W_1}
=\frac{1}{2}\norm{A}_1.
\end{eqnarray}
\end{lem}

\section{Basic properties of quantum Wasserstein complexity}\label{Appen:basic}
\begin{prop}[Restatement of Proposition~\ref{prop:bas_1}]
The quantum Wasserstein complexity $C_{W_1}$ satisfies the following properties:

\begin{enumerate}[(1)]
\item \label{item:restatement_faithfulness}
Faithfulness: $C_{W_1}(\Lambda)=0$ if and only if $\Lambda$ is the identity map;

\item \label{item:restatement_convexity}
Convexity: $C_{W_1}\left(\sum_ip_i\Lambda_i\right)\leq \sum_ip_iC_{W_1}(\Lambda_i)$, where 
$p_i\geq 0$ and $\sum_ip_i=1$;

\item \label{item:restatement_subadditivity}
Subadditivity under concatenation:  $C_{W_1}(\Lambda_1\circ \Lambda_2)\leq C_{W_1}(\Lambda_1)+C_{W_1}(\Lambda_2)$;

\item \label{item:restatement_tensor_sub}
$C_{W_1}(\Lambda_1\ot \Lambda_2)\leq C_{W_1}(\Lambda_1\ot I)+C_{W_1}(\Lambda_2\ot I)$;

\item \label{item:restatement_tensor_sup}
Superadditivity under tensorization: $C_{W_1}(\Lambda_1\ot \Lambda_2)\geq C_{W_1}(\Lambda_1)+C_{W_1}(\Lambda_2)$;

\item \label{item:restatement_unitary}
For a unitary channel $U$, $C_{W_1}(U)=C_{W_1}(U^\dag)$;

\item \label{item:restatement_nontrivially}
If $\Lambda$ acts nontrivially on a $k$-qudit subsystem, then we have 
$C_{W_1}(\Lambda)\leq k$.

\end{enumerate}
\end{prop}
\begin{proof}
Property \ref{item:restatement_faithfulness} follows directly from the faithfulness of the quantum $W_1$ norm, properties \ref{item:restatement_convexity}--\ref{item:restatement_tensor_sub} follow directly from the 
triangle inequality of the quantum $W_1$ norm, and property \ref{item:restatement_tensor_sup} follows from the tensorization of
$W_1$ norm (see Lemma  \ref{lem:ten}) as follows
\begin{eqnarray*}
\norm{\rho-\Lambda_1\ot\Lambda_2(\rho)}_{W_1}
\geq \norm{\rho_1-\Lambda_1(\rho_1)}_{W_1}
+\norm{\rho_2-\Lambda_2(\rho_2)}_{W_1},
\end{eqnarray*}
where $\rho_i$ (for $i=1,2$)  denotes the corresponding reduced state.
Property \ref{item:restatement_unitary} follows directly from the definition of the quantum Wasserstein complexity. Finally, Property \ref{item:restatement_nontrivially}
follows from Lemma~\ref{lem:boun_1norm}.

\end{proof}

%
%\begin{prop}
%Quantum Wasserstein complexity $W_1(\Lambda) $ satisfies the following properties:
%
%(1) Faithfulness: $W_1(\Lambda)=0$ if and only if $\Lambda$ is the identity map
%
%(2) Convexity: $W_1(\sum_ip_i\Lambda_i)\leq \sum_ip_iW_1(\Lambda_i)$, where 
%$p_i\geq 0$ and $\sum_ip_i=1$. 
%
%(3) Subadditivity under concatenation:  $W_1(\Lambda_1\circ \Lambda_2)\leq W_1(\Lambda_1)+W_1(\Lambda_2)$.
%
%(4) $W_1(\Lambda_1\ot \Lambda_2)\leq W_1(\Lambda_1\ot I)+W_1(\Lambda_2\ot I)$.
%
%(5) $W_1(\Lambda_1\ot \Lambda_2)\geq W_1(\Lambda_1)+W_1(\Lambda_2)$.
%
%(6) For unitary channel $U$, $W_1(U)=W_1(U^\dag)$.
%\end{prop}
%
%\begin{proof}
%(1) comes directly from the faithfulness of the quantum $W_1$ norm, (2-4) comes directly from the 
%triangle inequality of the quantum $W_1$ norm, and property (5) comes from the tensorization  of
%$W_1$ norm ( See Lemma  \ref{lem:ten}) as follows
%\begin{eqnarray*}
%\norm{\rho-\Lambda_1\ot\Lambda_2(\rho)}_{W_1}
%\geq \norm{\rho_1-\Lambda_1(\rho_1)}_{W_1}
%+\norm{\rho_2-\Lambda_2(\rho_2)}_{W_1},
%\end{eqnarray*}
%where $\rho_i$ ($i=1,2$)  denotes the corresponding reduced state.
%
%\end{proof}

\begin{prop}[Restatement of Proposition~\ref{prop:bas_2}]
The correlation-assisted Wasserstein complexity $\AC_{W_1}(\Lambda)$ 
satisfies the following properties:
\begin{enumerate}[(1)]

\item \label{item:faithfulness}
Faithfulness: $\AC_{W_1}(\Lambda)=0$ if and only if $\Lambda$ is the identity map;

\item \label{item:convexity}
Convexity: $\AC_{W_1}(\sum_ip_i\Lambda_i)\leq \sum_ip_i\AC_{W_1}(\Lambda_i)$, where 
$p_i\geq 0$ and $\sum_ip_i=1$;

\item \label{item:Subadditivity_under_concatenation}
Subadditivity under concatenation:  $\AC_{W_1}(\Lambda_1\circ \Lambda_2)\leq \AC_{W_1}(\Lambda_1)+\AC_{W_1}(\Lambda_2)$;

\item \label{item:additivity_under_tensorization} 
Additivity under tensorization: $\AC_{W_1}(\Lambda_1\ot \Lambda_2)=\AC_{W_1}(\Lambda_1)+\AC_{W_1}(\Lambda_2)$.
\end{enumerate}
\end{prop}
\begin{proof}
The property \ref{item:faithfulness} follows from the faithfulness of the  quantum $W_1$ distance. Properties \ref{item:convexity} and \ref{item:Subadditivity_under_concatenation} follow from the 
triangle inequality of the quantum $W_1$ distance. Hence, it suffices for us to prove the property \ref{item:additivity_under_tensorization}.
Since 
\begin{eqnarray}
C_{W_1}(\Lambda_1\ot \Lambda_2)
\leq C_{W_1}(\Lambda_1\ot I)
+C_{W_1}(\Lambda_2\ot I),
\end{eqnarray}
by the definition of $EW_1$, we have 
\begin{eqnarray}
\AC_{W_1}(\Lambda_1\ot \Lambda_2)
\leq \AC_{W_1}(\Lambda_1)
+\AC_{W_1}(\Lambda_2).
\end{eqnarray}
Besides, since 
\begin{eqnarray}
C_{W_1}(\Lambda_1\ot \Lambda_2)
\geq C_{W_1}(\Lambda_1)
+C_{W_1}(\Lambda_2),
\end{eqnarray}
we have
\begin{eqnarray}
\AC_{W_1}(\Lambda_1\ot \Lambda_2)
\geq \AC_{W_1}(\Lambda_1)
+\AC_{W_1}(\Lambda_2).
\end{eqnarray}
\end{proof}

\section{Estimation of \texorpdfstring{$C_{W_1}$}{CW1} and \texorpdfstring{$\AC_{W_1}$}{ACW1}: interesting examples}\label{appen:exam}

Let us start with the simplest case, where $U$ is a single-qudit unitary.

\begin{claim}
For a 1-qudit unitary $U$, we have 
\begin{eqnarray}
C_{W_1}(U)=\AC_{W_1}(U).
\end{eqnarray}
\end{claim}
\begin{proof}
For any 1-qudit unitary and pure state $\ket{\psi}$, we have
\begin{eqnarray}
\norm{\proj{\psi}-U\proj{\psi}U^\dag}_{W_1}
=\frac{1}{2}
\norm{\proj{\psi}-U\proj{\psi}U^\dag}_{1}
=\sqrt{1-|\bra{\psi}U\ket{\psi}|^2}.
\end{eqnarray}
Let $U=\sum_j\exp(i\theta_j)\proj{a_j}$ be the eigenvalue decomposition of $U$. Then,
the state $\ket{\psi}$ can be written as
$\ket{\psi}=\sum_jc_j\ket{a_j}$, from which it follows that
$|\bra{\psi}U\ket{\psi}|^2=|\sum_{j}|c_j|^2\exp(i\theta_j)|^2=|\sum_jd_j\exp(i\theta_j)|^2$, where $d_j = |c_j|^2$.
Hence 
\begin{eqnarray*}
C_{W_1}(U)&=&
\max_{\psi}\norm{\proj{\psi}-U\proj{\psi}}_{W_1}
=\max_{\psi}
\sqrt{1-|\bra{\psi}U\ket{\psi}|^2}\\
&=&\sqrt{1-\min_{\psi}|\bra{\psi}U\ket{\psi}|^2}
=\sqrt{1-\min_{d_j:\sum_jd_j=1}\left|\sum_jd_j\exp(i\theta_j)\right|^2}.
\end{eqnarray*}
Let us consider $U\ot I_m$ for any $m\geq1$. For any ($m+1$)-qudit state
$\ket{\psi}=\sum_j\sqrt{\lambda_j}\ket{\psi_j}\ket{j}$, 
with $\ket{j}$ being a basis on $m$-qudit systems and $\sum_j\lambda_j=1$, we have
\begin{eqnarray*}
\bra{\psi}U\ot I\ket{\psi}
=\sum_j\lambda_j\bra{\psi_j}U\ket{\psi_j},
\end{eqnarray*}
where for each $j$, $\ket{\psi_j}=\sum_k\sqrt{d_{jk}}e^{i\phi_{jk}}\ket{a_k}$, and 
\begin{eqnarray*}
\bra{\psi_j}U\ket{\psi_j}
=\sum_{k}d_{jk}
\exp(i\theta_k).
\end{eqnarray*}
Hence
\begin{eqnarray*}
\bra{\psi}U\ot I\ket{\psi}
=\sum_{j,k}
\lambda_jd_{jk}\exp(i\theta_{k})
=\sum_{k}f_k
\exp(i\theta_k),
\end{eqnarray*}
where $f_k=\sum_j\lambda_jd_{jk}$, and $\sum_kf_k=1$. 
Hence, we have the following expression for $\AC_{W_1}(U)$:
\begin{eqnarray*}
\AC_{W_1}(U)
=\sup_{m}
\max_{\psi}
\norm{\proj{\psi}-U\ot I_m\proj{\psi}U^\dag\ot I_m}_{W_1}
=\sqrt{1-\min_{f_k:\sum_kf_k=1}\left|\sum_kf_k\exp(i\theta_k)\right|^2},
\end{eqnarray*}
which implies that
$\AC_{W_1}(U)=C_{W_1}(U)$.
\end{proof}

\subsection*{Example 1: Depolarizing channel}
Let us consider the single-qudit depolarizing channel, for which we have the following statement.

\begin{claim}
For the 1-qudit depolarizing channel $D_p(\cdot)=p(\cdot)+(1-p)\Tr{\cdot}I/d$, we have
\begin{eqnarray}
C_{W_1}(D_p)
=(1-p)(1-1/d),
\end{eqnarray}
and 
\begin{eqnarray}
\AC_{W_1}(D_p)
=(1-p)(1-1/d^2).
\end{eqnarray}
\end{claim}

\begin{proof}
For any 1-qudit state, we have
\begin{eqnarray*}
\norm{\proj{\psi}
-D_p(\proj{\psi})}_{W_1}
=\frac{1}{2}
\norm{\proj{\psi}
-D_p(\proj{\psi})}_1
=(1-p)(1-1/d).
\end{eqnarray*}
Hence, we have 
\begin{eqnarray*}
C_{W_1}(D_p)
=(1-p)(1-1/d).
\end{eqnarray*}

Now, let us take a 2-qudit state $\ket{\psi}=\frac{1}{\sqrt{d}}\sum_j\ket{j}\ket{j}$ and the quantum channel $D_p\ot I$. Then,

\begin{eqnarray*}
\norm{\proj{\psi}
-D_p\ot I(\proj{\psi})}_{W_1}
=\frac{1}{2}
\norm{\proj{\psi}
-D_p\ot I (\proj{\psi})}_1
=(1-p)(1-1/d^2).
\end{eqnarray*}
This implies that 
\begin{eqnarray*}
\AC_{W_1}
(D_p)\geq (1-p)(1-1/d^2).
\end{eqnarray*}
Besides, let us prove that 
for any integer $m\geq 0$, 
\begin{eqnarray*}
C_{W_1}(D_p\ot I_m)\leq (1-p)(1-1/d^2).
\end{eqnarray*}
For any ($m+1$)-qudit state $\ket{\psi}$ , we have 
\begin{eqnarray*}
&&\norm{\proj{\psi}-D_p\ot I_{m}(\proj{\psi})}_{W_1}\\
&=&(1-p)
\norm{\proj{\psi}-\frac{I}{d}\ot \Ptr{1}{\proj{\psi}}}_{W_1}\\
&=&(1-p)
\norm{\proj{\psi}-\frac{1}{d_2}\sum_{s,t\in Z_d}X^s_1Z^t_1\proj{\psi}Z^{-t}_1X^{-s}_1}_{W_1}\\
&=&(1-p)
\frac{d^2-1}{d^2}
\norm{\mathbb{E}_{(s,t)\neq (0,0)}(\proj{\psi}-X^s_1Z^t_1\proj{\psi}Z^{-t}_1X^{-s}_1)}_{W_1}\\
&\leq&(1-p)
\frac{d^2-1}{d^2}
\mathbb{E}_{(s,t)\neq (0,0)}
\norm{\proj{\psi}-X^s_1Z^t_1\proj{\psi}Z^{-t}_1X^{-s}_1}_{W_1}\\
&=&(1-p)
\frac{d^2-1}{d^2}\frac{1}{2}\norm{\proj{\psi}-X^s_1Z^t_1\proj{\psi}Z^{-t}_1X^{-s}_1}_{1}\\
&\leq&(1-p)
\frac{d^2-1}{d^2},
\end{eqnarray*}
where the third line comes from the fact that 
\begin{eqnarray*}
\frac{1}{d^2}\sum_{(s,t)\in \mathbb{Z}_d\times \mathbb{Z}_d}X^sZ^t(\cdot)Z^{-t}X^{-s}=\Tr{\cdot}\frac{I}{d},
\end{eqnarray*}
the fifth line comes from the convexity of the quantum $W_1$ distance, and the last line 
comes from the fact that $\Ptr{1}{\proj{\psi}-X^s_1Z^t_1\proj{\psi}Z^{-t}_1X^{-s}_1}=0$, 
and thus $\norm{\proj{\psi}-X^s_1Z^t_1\proj{\psi}Z^{-t}_1X^{-s}_1}_{W_1}= \frac{1}{2}\norm{\proj{\psi}-X^s_1Z^t_1\proj{\psi}Z^{-t}_1X^{-s}_1}_{1}$.
Thus, we have 
\begin{eqnarray*}
\AC_{W_1}(D_p)\leq (1-p)(1-1/d^2).
\end{eqnarray*}

\end{proof}

\begin{cor}
For the $n$-fold tensor product of the depolarizing channel $D^{\ot n}_p$, we have 
\begin{eqnarray}
n(1-p)(1-1/d)\leq C_{W_1}(D^{\ot n}_p)\leq n(1-p)(1-1/d^2)
\end{eqnarray}
and 
\begin{eqnarray}
\AC_{W_1}(D^{\ot n}_p)= n(1-p)(1-1/d^2).
\end{eqnarray}
\end{cor}

\subsection*{Example 2: \texorpdfstring{$n$}{n}-fold tensor product of Hadamard gates}
% Example 2 (Tensor product of Hadamard gates): 
Let us consider the $n$-fold tensor product of Hadamard gates, i.e., 
$H^{\ot n}$. The quantum 
Wasserstein complexity of $H^{\ot n}$ is 
\begin{eqnarray}
\label{qW_of_H}
C_{W_1}(H^{\ot n})
=n.
\end{eqnarray}
\begin{proof}[Proof of \eqref{qW_of_H}]
The inequality $C_{W_1}(H^{\ot n})\leq n$ holds because
$\norm{\rho-\sigma}_{W_1}\leq \frac{n}{2}\norm{\rho-\sigma}_1\leq n$, which follows from Lemma~\ref{lem:boun_1norm}. To prove that $C_{W_1}(H^{\ot n})\geq n$, we use the fact that for the Hadamard gate $H$, there exists some pure state $\ket{\phi}$ such that $H\ket{\phi}\bot \ket{\phi}$. Hence, 
\begin{eqnarray*}
C_{W_1}(H^{\ot n})
&\geq& \norm{(H\proj{\phi}H)^{\ot n}-\proj{\phi}^{\ot n}}_{W_1}\\
&=&\sum_i \norm{H\proj{\phi}H-\proj{\phi}}_{W_1}\\
&=&\frac{1}{2}\sum_i \norm{H\proj{\phi}H-\proj{\phi}}_{1}\\
&=&n,
\end{eqnarray*}
where the first line follows from the definition of $C_{W_1}$; the second line follows from the tensorization of the quantum $W_1$ distance, i.e., $\norm{\rho_1\ot\rho_2-\sigma_1\ot\sigma_2}_{W_1}=
\norm{\rho_1-\sigma_1}_{W_1}+\norm{\rho_2-\sigma_2}_{W_1}$ \cite{PalmaIEEE2021}; the third line follows from the fact that $\norm{\rho-\sigma}_{W_1}=\frac{1}{2}\norm{\rho_1-\sigma_1}_1$
 for 1-qudit states $\rho, \sigma$ \cite{PalmaIEEE2021}; and the last line follows from the fact that  $\norm{\rho_1-\sigma_1}_1=1$ as
 $H\ket{\phi}\bot \ket{\phi}$.
%\end{itemize}
\end{proof}

Besides, as $ C_{W_1}(H^{\ot n}\ot I_m)\leq n$, we also have 
 \begin{eqnarray}
\AC_{W_1}(H^{\ot n})=n.
 \end{eqnarray}

\subsection*{Example 3: CNOT gates}
%Example 3 (CNOT gates): 
Let us consider the $n$-qubit quantum circuit $\prod_i \CNOT_{i,i+1}$ (See Fig.~\ref{fig0}). Then we have
\begin{eqnarray}
 C_{W_1}\left(\prod_i \CNOT_{i,i+1}\right)=\Theta(n),
\end{eqnarray}
where $f(n)=\Theta(n)$ means that there exist constants $c_1,c_2$ such that $c_1n\leq f(n)\leq c_2 n$ for any $n$.
(1) The upper bound $C_{W_1}\left(\prod_i \CNOT_{i,i+1}\right)\leq n$ holds for the same reason as the one in the above example.
(2) Next, we show that $ C_{W_1}(\prod_i \CNOT_{i,i+1})\geq n/2$. To this end, let us take the input state to be 
$\ket{+}\ot \ket{0}^{n-1}$. Then, the output state is $\prod_i \CNOT_{i,i+1}\ket{+}\ot \ket{0}^{n-1}=\frac{1}{\sqrt{2}}(\ket{0}^n+\ket{1}^n)$, which we denote as $\ket{\Cat_{1/2}}$.
Hence 
\begin{eqnarray*}
C_{W_1}\left(\prod_i \CNOT_{i,i+1}\right)
&\geq&\norm{\proj{+}\ot \proj{0}^{n-1}-\proj{\Cat_{1/2}}}_{W_1}\\
&\geq& \frac{1}{2}\sum_i\norm{\rho_i-\sigma_i}_1=n/2,
\end{eqnarray*}
where $\rho_i$ (or $\sigma_i$) is the 
reduced state of $\proj{+}\ot \proj{0}^{n-1}$ (or $\proj{\Cat_{1/2}}$) on the $i$-th position, and the second inequality comes from the tensorization of the quantum $W_1$ distance.

Hence, we also have 
\begin{eqnarray}
\AC_{W_1}\left(\prod_i \CNOT_{i,i+1}\right)=\Theta(n).
\end{eqnarray}

\begin{figure}[t]
  \center{\includegraphics[width=5cm]  {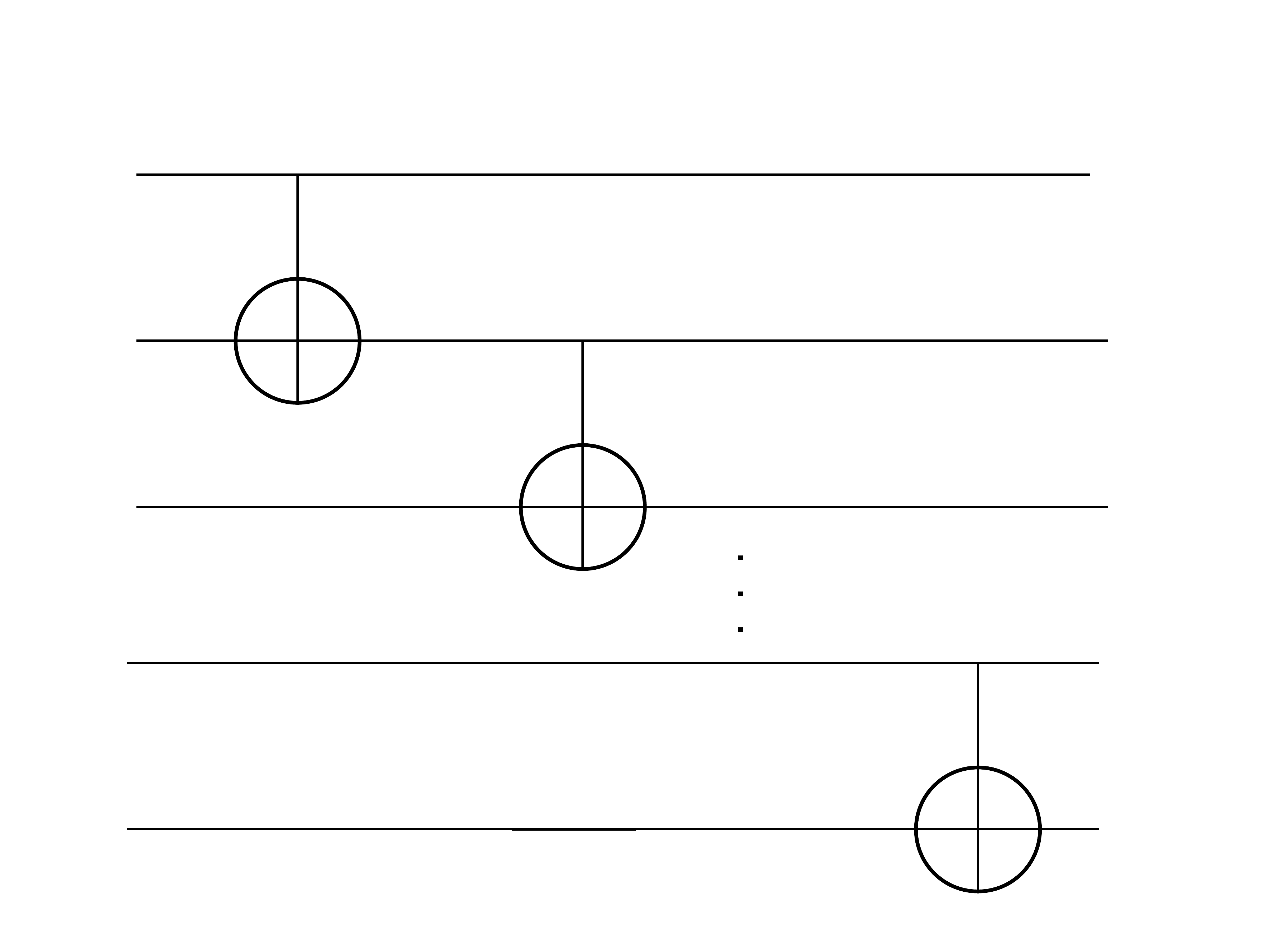}}     
  \caption{
 A circuit diagram for the circuit of cascading CNOT gates $\prod_i\CNOT_{i,i+1}$.}
  \label{fig0}
 \end{figure}

\section{Application in quantum circuit complexity and experimental cost}\label{Appen:appl}

\begin{lem}\label{lem:rate}
Given an $n$-qudit system with a Hamiltonian $H$ acting on a $k$-qudit subsystem, and an $n$-qudit state 
$\ket{\psi}$, for the unitary $U_{\Delta t}=\exp(-iH \Delta t)$, one has
\begin{eqnarray}
\norm{\proj{\psi}-U_{\Delta t}\proj{\psi} U^\dag_{\Delta t}}_{W_1}
\leq 2\sqrt{2}k\frac{d^2-1}{d^2}
\norm{H}_{\infty}|\Delta t|e^{\norm{H}_{\infty}|\Delta t|},
\end{eqnarray}
which implies that 
\begin{eqnarray}
R_{W_1}(\psi,H)
\leq 2\sqrt{2}\norm{H}_{\infty}.
\end{eqnarray}
\end{lem}
\begin{proof}
Since $H$ acts on a $k$-qudit subsystem, there exists a subsystem $S$ for which $\Ptr{S}{\proj{\psi}-U_{\Delta t}\proj{\psi}U^\dag_{\Delta t}}=0$.
Hence, by Lemma \ref{lem:boun_1norm}, we have 
\begin{eqnarray*}
\norm{\proj{\psi}-U_{\Delta t}\proj{\psi} U^\dag_{\Delta t}}_{W_1}
\leq  k\frac{d^2-1}{d^2}
\norm{\proj{\psi}-U_{\Delta t}\proj{\psi} U^\dag_{\Delta t}}_{1}
= 2 k\frac{d^2-1}{d^2}
\sqrt{1-|\bra{\psi}U_{\Delta t}\ket{\psi}|^2},
\end{eqnarray*}
Let us define $\ket{\psi_{\Delta t}}=U_{\Delta t}\ket{\psi}$. Then, the 
Taylor expansion of $\proj{\psi_{\Delta t}}$ is 
\begin{eqnarray*}
\proj{\psi_{\Delta t}}
=\proj{\psi}
+i\Delta t[H,\proj{\psi}]
+\frac{(i\Delta t)^2}{2!}[H,[H,\proj{\psi}]]
+ \cdots
\end{eqnarray*}
where $[A,B]=AB-BA$ denotes the commutator between $A$ and $B$.
Then
\begin{eqnarray*}
1-|\bra{\psi}U_{\Delta t}\ket{\psi}|^2
&=&1-\left[\Tr{\proj{\psi}^2}
+i\Delta t\Tr{\proj{\psi}[H,\proj{\psi}]}
+\frac{(i\Delta t)^2}{2!}\Tr{\proj{\psi}[H,[H,\proj{\psi}]]}
+\cdots\right]\\
&=&
-\frac{(i\Delta t)^2}{2!}\Tr{\proj{\psi}[H,[H,\proj{\psi}]]}+\cdots
\end{eqnarray*}
where $\Tr{\proj{\psi}[H,\proj{\psi}]}=0$. 
Let us define $\ad_H$ as $\ad_H(A):=[H, A]$.
Then, the above formula can be rewritten as 
\begin{eqnarray*}
1-|\bra{\psi}U_{\Delta t}\ket{\psi}|^2
=-\sum_{k\geq 2}
\frac{(i\Delta t)^k}{k!}
\Tr{\proj{\psi}\ad^k_H(\proj{\psi})}.
\end{eqnarray*}
For each term $\Tr{\proj{\psi}\ad^k_H(\proj{\psi})}$, by H\"older's inequality, we have 
\begin{eqnarray*}
\left|\Tr{\proj{\psi}\ad^k_H(\proj{\psi})}\right|\leq 
\norm{\ad^k_H(\proj{\psi})}_{\infty}
\leq \norm{H}^k_{\infty}2^k,
\end{eqnarray*}
which implies that
\begin{eqnarray}
\left|1-|\bra{\psi}U_{\Delta t}\ket{\psi}|^2\right|
\leq 
2\norm{H}^2_{\infty}
|\Delta t|^2
\exp(2\norm{H}_{\infty}|\Delta t|).
\end{eqnarray}
Hence, we have 
\begin{eqnarray}
 \norm{\proj{\psi}-U_{\Delta t}\proj{\psi} U^\dag_{\Delta t}}_{W_1}
\leq 2\sqrt{2}k\frac{d^2-1}{d^2}
\norm{H}_{\infty}|\Delta t|e^{\norm{H}_{\infty}|\Delta t|}.
\end{eqnarray}
\end{proof}

\begin{thm}
Given a unitary $U$, the circuit cost $\Cost(U)$ is lower bounded in terms of
$C_{W_1}(U)$ as follows
\begin{eqnarray}
\Cost(U)\geq 4\sqrt{2}C_{W_1}(U).
\end{eqnarray}
\end{thm}
\begin{proof}
First, let us take 
a Trotter decomposition of $U$ such that for arbitrarily small $\epsilon>0$,
\begin{eqnarray*}
\norm{U-V_N}_{\infty}
\leq \epsilon,
\end{eqnarray*}
where $V_N$ is defined as 
follows
\begin{eqnarray*}
V_N&:=&\prod^N_{t=1}W_t,\\
W_t&:=&\exp\left(-\frac{i}{N}\sum^m_{j=1}r_j\left(\frac{t}{N}\right)h_j\right).
\end{eqnarray*}
and 
\begin{eqnarray*}
W_t&=&\lim_{l\to\infty}W^{(l)}_t,\\
W^{(l)}_t&:=&\left(W^{1/l}_{t,1} \cdots W^{1/l}_{t,m}\right)^l,\\
W_{t,j}&:=&\exp\left(
-\frac{i}{N}
r_{j}\left(
\frac{t}{N}
\right)h_j
\right).
\end{eqnarray*}
where $\norm{h_j}_{\infty} \leq 1$ for any $j$.
Let us define 
$\ket{\psi_t}=
W_t\ket{\psi_{t-1}}$ with $\ket{\psi_0}=\ket{\psi}$,
then by the triangle inequality of the quantum $W_1$ distance,  we have 
\begin{eqnarray}
\norm{\proj{\psi}-V_N\proj{\psi} V^\dag_N}_{W_1}
\leq 
\sum^N_{t=1}
\norm{\proj{\psi_{t-1}}-\proj{\psi_t}}_{W_1}.
\end{eqnarray}
For each $\norm{\proj{\psi_{t-1}}-\proj{\psi_t}}_{W_1}$, we have 
\begin{eqnarray*}
\norm{\proj{\psi_{t-1}}-\proj{\psi_t}}_{W_1}
&=&\norm{\proj{\psi_{t-1}}-W_t\proj{\psi_{t-1}}W^\dag_t}_{W_1}\\
&=&\lim_{l\to \infty}
\norm{\proj{\psi_{t-1}}-W^{(l)}_t\proj{\psi_{t-1}}W^{(l),\dag}_t}_{W_1}\\
&\leq&  \lim_{l\to \infty}l\sum_j\frac{4\sqrt{2}}{Nl}\left|r_{j}\left(
\frac{t}{N}
\right)\right|
\exp\left(\frac{1}{Nl}\left|r_{j}\left(
\frac{t}{N}
\right)\right|\right)\\
&=&\frac{4\sqrt{2}}{N}
\sum_j
\left|r_{j}\left(
\frac{t}{N}
\right)\right|,
\end{eqnarray*}
where the last inequality comes from the triangle inequality of quantum 
$W_1$ distance and Lemma~\ref{lem:rate} by taking $\Delta t=\frac{1}{l}$.
Therefore, 
\begin{eqnarray}
\norm{\proj{\psi}-V_N\proj{\psi} V^\dag_N}_{W_1}
\leq\frac{4\sqrt{2}}{N}
\sum^N_{t=1}
\sum^m_{j=1}
\left|r_{j}\left(
\frac{t}{N}
\right)\right|.
\end{eqnarray}
Since the circuit cost can be expressed as
 \begin{eqnarray*}
 \Cost(U)=
 \lim_{N\to \infty}\frac{1}{N}\sum^N_{t=1}
\sum^m_{j=1}
\left|r_j\left(\frac{t}{N}\right)\right|,
 \end{eqnarray*}
we have 
\begin{eqnarray}
\Cost(U)
\geq 4\sqrt{2}C_{W_1}(U).
\end{eqnarray}
\end{proof}

Note that Lemma~\ref{lem:rate} also holds for ($n+m$)-qudit systems, and hence the statement also holds for 
$\AC_{W_1}$. Therefore, the following corollary follows immediately.
 \begin{cor}
Given a unitary $U$, the circuit cost $\Cost(U)$ is lower bounded in terms of
$\AC_{W_1}(U)$ as follows:
\begin{eqnarray}
\Cost(U)\geq 4\sqrt{2}\AC_{W_1}(U).
\end{eqnarray}
\end{cor}

\begin{thm}
The experimental cost $\mathcal{R}_U$ of an $n$-qubit quantum circuit $U=\prod_lU_l$
is lower bounded in terms of $C_{W_1}(U)$ as follows:
\begin{eqnarray}
\mathcal{R}(U)
\geq \frac{1}{2}C_{W_1}(U).
\end{eqnarray}
\end{thm}
\begin{proof}
For any pure $n$-qubit quantum state $\ket{\psi}$, let us define 
\begin{eqnarray}
\ket{\psi_{l+1}}:=U_l\ket{\psi_l}, \quad \mbox{with } 
\ket{\psi_1}=\ket{\psi}. 
\end{eqnarray}
And we also define
\begin{eqnarray}
\ket{\psi_{l,t_l}}:=e^{iH_lt_l}\ket{\psi_l},
\end{eqnarray}
as the intermediate 
state at time $t_l\in [0,T_l]$ while implementing $U_l$
with $\ket{\psi_{l,0}}=\ket{\psi_l}, \psi_{l,T_l}=\ket{\psi_{l+1}}$.
 
Since $U_l$ acts on only $k_l$ qubits, by Lemma~\ref{lem:boun_1norm}, we have
\begin{eqnarray}\label{eq:w_11}
\norm{\proj{\psi_{l+1}}-\proj{\psi_l}}_{W_1}
\leq k_l\norm{\proj{\psi_{l+1}}-\proj{\psi_l}}_1.
\end{eqnarray}
Besides, it was shown in \cite{Deffner_2017} that the quantum speed limit for the
trace norm is bounded by the trace norm of the derivative of the state, that is, 
\begin{eqnarray}
\norm{\proj{\psi_{l+1}}-\proj{\psi_l}}_1
\leq \int^{T_l}_0\norm{d \proj{\psi_{l,t_l}}/dt_l}_1dt_l,
\end{eqnarray}
where ${d \proj{\psi_{l,t_l}}/dt_l}=d e^{iH_lt_l}\proj{\psi_l}e^{-iH_lt_l}/dt_l=i[H_l, \proj{\psi_{l,t_l}}]$.
The trace norm of the derivative of the state under unitary evolution $e^{iH_lt_l}$ is bounded above by the square root of  the variance of the Hamiltonian $H_l$, i.e.,
\begin{eqnarray}
\norm{d \proj{\psi_{l,t_l}}/dt_l}_1
\leq 2\sqrt{\Var_{\psi_{l,t_l}}[H_l]},
\end{eqnarray}
where 
\begin{eqnarray}
\Var_{\psi_{l,t_l}}[H_l]
=\Tr{H^2_l\proj{\psi_{l,t_l}}}
-\Tr{H\proj{\psi_{l,t_l}}}^2.
\end{eqnarray}
Note that $\Var_{\psi_l}[H_l]
\leq E^2_l$ with $E_l=(h_{2^{k_l}}-h_1)/2$, and so we have 
\begin{eqnarray}\label{eq:1_ET}
\norm{\proj{\psi_{l+1}}-\proj{\psi_l}}_1
\leq 2E_lT_l.
\end{eqnarray}
Hence, 
\begin{eqnarray}
\norm{\proj{\psi_{l+1}}-\proj{\psi_l}}_{W_1}
\leq 2k_lE_lT_l.
\end{eqnarray}
Therefore, for the quantum circuit $U=\prod_l U_l$ with any input state $\ket{\psi}$, 
\begin{eqnarray*}
\norm{\proj{\psi}-U\proj{\psi}U^\dag}_{W_1}
&\leq& \sum_l
\norm{\proj{\psi_{l+1}}-\proj{\psi_l}}_{W_1}\\
&\leq& \sum_l k_l
\norm{\proj{\psi_{l+1}}-\proj{\psi_l}}_1\\
&\leq& 2\sum_l k_l E_l T_l=2\mathcal{R}(U),
\end{eqnarray*}
where the first inequality comes from the triangle inequality, the second inequality comes from \eqref{eq:w_11}, and the third inequality comes from \eqref{eq:1_ET}.
\end{proof}

The proof given for the above theorem also works for the case where we have $m$ qubits as ancillas. Hence, the following corollary follows immediately.
\begin{cor}
The experimental cost of an $n$-qubit quantum circuit $\mathcal{R}(U)$ is lower bounded in terms of
$\AC_{W_1}(U)$ as follows: 
\begin{eqnarray}
\mathcal{R}(U)
\geq \frac{1}{2}\AC_{W_1}(U).
\end{eqnarray}
\end{cor}

\end{document}